%% file: main.tex
\documentclass[preprint,12pt,authoryear]{elsarticle}

\usepackage{amssymb}
\usepackage{amsmath}
 \usepackage{amsthm}

\usepackage{amsfonts}
\usepackage{stmaryrd}
\usepackage{xcolor}

\usepackage{ltlfonts}
\usepackage{wasysym}
\usepackage{verbatim}
\usepackage[ruled,vlined,linesnumbered]{algorithm2e}
\usepackage[normalem]{ulem}  
\usepackage{hyperref}

\input{commands.tex}    

\journal{Information and Computation}

\begin{document}

\begin{frontmatter}



\title{A Correct Algorithm for Identifying Independent Variable Sets in Reactive Systems} 
\tnotetext[t1]{This work is partially funded by the European Regional Development Fund (ERDF) through grant PID2020-112581GB-C22.}

\author[1]{Josu Oca}
\ead{josu.oca@udg.edu}

\affiliation[1]{organization={University of Girona},
city={Girona},
country={Spain}}

\author[2]{Montserrat Hermo}
\ead{montserrat.hermo@ehu.eus}

\affiliation[2]{organization={University of the Basque Country},
city={San Sebasti\'an},
country={Spain}}

\author[3]{Alexander Bolotov}
\ead{A.Bolotov@westminster.ac.uk}

\affiliation[3]{organization={University of Westminster},
city={London},
country={UK}}

\begin{abstract}
Recent work has proposed algorithms for decomposing reactive synthesis specifications into simpler and independent subspecifications. Motivated by the {\it DecomposeContract} algorithm introduced by Antonio Iannopollo, we revisit this approach and provide a mathematical account of the notion of independence on which it is based. The central idea in this setting is to identify independence among system-controlled variables in linear temporal logic formulae by exploiting the power of a model checker.

Although the original {\it DecomposeContract} algorithm is sound, it is not complete. We support this observation by presenting a concrete counterexample, and we then propose a refined decomposition procedure that preserves the model-checking-based nature of the original method while ensuring completeness. Beyond this algorithmic result, our main contribution is a rigorous semantic analysis of the method, which reveals the principles behind it and makes its limitations explicit.
\end{abstract}




\end{frontmatter}



\section{Introduction}
\label{sec:intro}

Linear Temporal Logic (\ltl), introduced in~\citep{P77}, is a standard formalism for specifying reactive systems. In the synthesis setting, atomic propositions are typically partitioned into variables controlled by the environment and variables controlled by the system. Given an \ltl specification, the realisability problem asks whether there exists an implementation that satisfies the specification against all possible behaviours of the environment, while synthesis constructs such an implementation whenever it exists. A major obstacle is that both problems are 2EXPTIME-complete~\citep{SC85,PR89}, which has motivated a wide range of techniques aimed at improving scalability.

A natural way to cope with this complexity is to decompose a specification into smaller subspecifications, solve them separately, and then combine the resulting partial solutions. Different forms of specification decomposition have been studied in the literature; we refer the reader to the {\it Related Work} below for an overview.

In this paper, we focus on the method introduced in~\citep{I18} and further developed in~\citep{I24}, in the setting of \ltl-based Assume/Guarantee contracts. This approach is particularly interesting because it proposes an efficient decomposition algorithm in which a model checker is used to identify sets of system-controlled variables that behave independently within the original \ltl formula. This makes the method especially attractive in practice, since modern model checkers such as NuSMV\footnote{\url{https://nusmv.fbk.eu/}} are highly efficient. The procedure runs in polynomial time, excluding the cost of the queries to the model checker.

We revisit this approach and introduce several changes that preserve its efficiency while ensuring completeness. These changes were announced in the conference poster~\citep{J24}; here we develop them fully and prove that the sets of independent variables returned by our method are minimal, in the sense that no non-empty proper subset of such a set is itself independent. Our contribution is primarily theoretical: we provide precise semantics for several notions that appear in~\citep{I18,I24}, and we prove in detail the properties required to place the decomposition method on a rigorous mathematical foundation.

\subsection*{Related Work}
\label{sec:related}

The idea of splitting a specification into smaller subspecifications and solving them separately already appears in compositional verification and synthesis~\citep{CLM89,RLP98}, and has since evolved in several directions. Recent work includes both theoretical approaches~\citep{K06,EK14} and tool-supported methods designed to handle specifications consisting of many conjuncts~\citep{M18,B20,G21}. One notable example is Lisa~\citep{B20}, which shows strong scalability on large conjunctions of \ltl formulae over finite traces~\citep{G13}. Its successor Lydia~\citep{G21} extends this capability by also supporting important classes of disjunctive specifications.

Many compositional techniques have been developed for prominent fragments of \ltl. These include methods for safety \ltl~\citep{KV01} and for the fragment of GR(1)~\citep{PPS06,BJPPS12,DM17,G25}. For full \ltl synthesis, particularly relevant compositional algorithms are given in~\citep{FJR10}.

Another line of work, presented in~\citep{KF18}, adopts a game-theoretic framework. In that approach, a specification is divided into two subspecifications, and winning regions for the corresponding subgames are computed in parallel. This idea was later extended in~\citep{IOSHY21}, where the decomposition is no longer restricted to two parts and can involve an arbitrary number of subspecifications. A different decomposition strategy, based on composing the solution sequentially rather than in parallel, is presented in~\citep{C23}.

The closest related contribution to the approach considered here is the work due to~\citep{FGP21}, which is formulated with enough precision to establish both soundness and completeness of the decomposition method. The authors introduce the notion of independent sublanguages, which provides an abstract mathematical characterisation of independence between subspecifications. They then instantiate this notion in the setting of nondeterministic B\"uchi automata and formulate the decomposition procedure at the automata level. While this automata-theoretic treatment is exact, it becomes impractical because it relies on expensive operations such as automata complementation~\citep{S88}. For this reason, the authors also propose an approximate algorithm that performs well in practice, although it does not necessarily compute an {\it optimal} decomposition.

\section{Preliminaries}
\label{sec:preliminaries}
\subsection*{ {\bf Linear Temporal Logic}}

\vspace{2mm}

The Linear Temporal Logic (\ltl)~\citep{P77} extends
propositional logic by temporal operators $\nx$ (next)
and $\until$ (until). Let  $\prop$ be a set of propositional variables.
%
 \ltl formulae are interpreted over traces. 
A {\it trace}~$\sigma$ is an infinite sequence of states $\sigma_0,\sigma_1,\dots$ where each state $\sigma_i\in2^\prop$. Intuitively, $\sigma_i$ represents the propositional variables that are true at the time point~$i$. 
Given a trace $\sigma = \sigma_0,\sigma_1,\sigma_2,\dots$, by $\sigma^i$ we denote the trace $\sigma_i,\sigma_{i+1},\dots$, where $i\geq 0$.
Given a trace $\sigma$ and an \ltl formula $\varphi$, below we define when $\sigma$ is a model of $\varphi$, for $\varphi$ constructed from the minimal syntax, where $\cT$ and $\cF$ are the standard abbreviations of constants true and false, the connectives $\vee,\to$ and $\Iff$ are defined classically, and the following abbreviations are used for the temporal operators: 
\begin{description}
\item 
$\quad\qquad$ eventually: $\cT\until\varphi$ for $\ev \varphi. \qquad$  always:  $\Not(\cT\until\Not\varphi)$ for $\alw\varphi$.
\item 
$\quad\qquad$ releases: $\neg(\neg\varphi\until\neg\psi)$  for $\varphi\releases\psi$.
\end{description}
We now inductively define when a trace $\sigma$ is a model of an \ltl formula $\varphi$, denoted by $\sigma \models \varphi$.
\[
\begin{array}{lcl}
\sigma \models p  &\mbox{ iff } & p \in \sigma_0\\
\sigma\models\neg\varphi &\text{iff} & \sigma\not\models\varphi\\
\sigma\models \varphi\wedge\psi &\text{iff}& \sigma\models\varphi \mbox{ and } \sigma\models\psi \\
\sigma \models\nx \varphi    &\mbox{ iff }& 
\sigma^{1} \models \varphi\\
 \sigma\models \varphi\until\psi  &\mbox{ iff } &
  \sigma^j \models \psi  \textrm{ for some } j  \textrm{ such that }  0 \leq j \textrm{ and } \\ 
  & & \sigma^i \models \varphi  \textrm{ for all } i  \textrm{ such that }  0 \leq i < j.
\end{array}
\]

If $\sigma\models \varphi$ then  we say that $\sigma$ models (or is a model of) $\varphi$. Referring by~$\Sigma_{\varphi}$ to the set of all traces that are models of $\varphi$, we say that the formula $\varphi$ is satisfiable if and only if  $\Sigma_{\varphi}$ is non-empty, otherwise $\varphi$ is unsatisfiable.  A formula $\varphi$ is valid if for every trace $\sigma$, $\sigma \in\Sigma_{\varphi}$.  
Two \ltl formulae $\varphi$ and $\psi$ are logically equivalent when $\Sigma_{\varphi} = \Sigma_{\psi}$.

The satisfiability problem for \ltl is PSPACE-complete and it was shown by \citep{SC85}.

It is well known that any satisfiable formula $\varphi$ has a model that is ultimately periodic. 
In our notation it means that $\varphi$ has a model $\sigma$ such that for some $j\geq i\geq 0$, the finite sequence $\sigma_i,\dots,\sigma_j$ is infinitely repeated, i.e.,  
$$\sigma = \sigma_0, \dots ,\sigma_{i-1}, (\sigma_i, \dots, \sigma_j)^\omega.$$
where $(\sigma_i,\dots,\sigma_j)^\omega$ denotes the infinite repetition of $\sigma_i,\dots,\sigma_j$.
\begin{example}
\label{ex:initial-running}
Consider the \ltl formula 
$$\varphi = q \wedge  \alw(\neg q \vee \nx q) \wedge ((p\wedge a)\until\neg q) \wedge \alw\ev\neg a.$$ 
The subformula $q \wedge \alw(\neg q \vee \nx q) \wedge (p\wedge a)\until\neg q $ is unsatisfiable, because~$q$ is required to hold initially, and the formula $\alw(\neg q \vee \nx q)$ forces it to persist forever. Indeed, whenever $q$ holds at some time point, it must also hold at the next state. Hence $\neg q$ can never become true, which contradicts the requirement $(p \wedge a)\until \neg q$. By contrast, the subformula $ q \wedge ((p\wedge a)\until\neg q) \wedge \alw\ev\neg a$  does not contain the persistence condition on~$q$. Therefore, $q$ may hold initially and become false at the next state. This makes the trace $\sigma = \sigma_0,(\sigma_1)^\omega$, with $\sigma_0=\{q,p,a\}$ and $\sigma_1=\emptyset$, a model of the last subformula.
\end{example}

\subsection*{ {\bf Realisability and synthesis}} 

\vspace{2mm}

Reactive systems~\citep{PR89} can be specified in the language of \ltl. 
 For the purpose of specification, the set of propositional variables, $\prop$, is divided into two disjoint subsets: $\E$, controlled by the environment, and $\Y$, controlled by the system.   
The behaviour of a reactive system can be modelled as a two-player game between the environment and the system, the players repeatedly choose Boolean values for the variables in~$\E$ and~$\Y$, respectively, over an infinite sequence of rounds. If we assume that the environment plays first, a winning strategy for the system can be defined as follows.

\begin{definition}[Winning strategy]
\label{def:winstr}
Let $\varphi$ be an \ltl formula over the variables in $\E \cup \Y$. A 
winning strategy (of the system) for $\varphi$ is a function $\theta : (2^\E)^+~\to~2^\Y$,
such that for every infinite string, usually called an environment play, 
$E = E_0 \cdot E_1 \cdots \in  (2^\E)^\omega$, the induced infinite trace 
$$\sigma^{\theta, E} = E_0 \cup \theta(E_0), E_1 \cup \theta(E_0 \cdot E_1),  E_2 \cup \theta(E_0 \cdot E_1 \cdot E_2), \cdots$$
is a model of $\varphi$.
\end{definition} 
 \noindent Thus, $\sigma^{\theta, E}$ is an infinite trace such that, for every state $i \geq 0$, 
 $$\sigma_i^{\theta, E} = E_i \cup \theta(E_0 \cdot E_1 \cdots E_i).$$  
 Note that $E_i \subseteq \E$ and $\theta(E_0 \cdot E_1 \cdots E_i) \subseteq \Y$.
\begin{example}[Environment plays and winning strategies]
\label{ex:numberzero}
Let $\E = \{p\}$ and $\Y = \{a,b\}$. Let $\varphi$ be the following \ltl formula:
$$
\varphi = \alw( (\neg p \to \alw \neg b ) \wedge (p \to (\alw a \, \vee \, \alw b) )).
$$
Suppose that $\theta$ is a system strategy such that, for every environment play
$
E = E_0 \cdot E_1 \cdot E_2 \cdots
$
and every $i \geq 0$,
$
\theta(E_0 \cdot E_1 \cdots E_i)=\{a\}.
$
Then $\theta$ is a winning strategy for $\varphi$. Indeed, $\theta$ always makes $a$ true and $b$ false, independently of the environment play. Hence, for every trace generated by~$\theta$, the formula
$
(\neg p \to \alw \neg b) 
$
holds at every time point, since $b$ is always false. Likewise,
$
(p \to (\alw a \, \vee \, \alw b))
$
holds at every state, since $a$ is always true.  Therefore, for every environment play $E$, we have
$
\sigma^{\theta,E} \models \varphi.
$

By contrast, a strategy that responds to the occurrence of $p$ by enforcing $\alw b$ cannot be winning. Indeed, if the environment later switches to $\neg p$, then the formula requires $\alw \neg b$ from that point onward. Thus, the obligations imposed by the two cases are incompatible, and no such strategy can satisfy $\varphi$ against all environment plays.
\end{example}
\begin{definition}[Realisability and synthesis]
An \ltl formula is realisable if there exists a winning strategy for it. Synthesis is the process of automatically constructing such a winning strategy whenever the formula is realisable.
\end{definition}

Unfortunately, both problems are 2EXPTIME-complete~\citep{PR89}, which makes them difficult to solve in practice.

\subsection*{{\bf Core idea behind Assume/Guarantee Contracts}} 

\vspace{1mm}
Contract-based design has been widely studied as a principled approach to the compositional development and verification of complex systems. In this setting, Assume/Guarantee reasoning provides a formal framework for decomposing system-level properties into component-level obligations, thereby enabling scalable verification. In the present paper, each component is viewed as a reactive system and is specified by an Assume/Guarantee contract.
 Such a contract consists of a pair $(A,G)$, where $A$, the assumption, is an \ltl formula describing the expected behaviour of the environment, and $G$, the guarantee, is an \ltl formula describing the behaviour that the system must ensure whenever the environment satisfies $A$. The use of such contracts~\citep{AL95,Benv18} supports modularity, because components can be analysed independently, compositionality, because components can be verified and synthesised separately, and clarity, because they explicitly distinguish what the system expects from what it promises.

An $(A,G)$ contract is said to be realisable if the system has a strategy that ensures $G$ on every play in which the environment satisfies $A$. In the present setting, deciding  realisability of an $(A,G)$ contract reduces to deciding realisability of an arbitrary \ltl formula. More precisely, the contract $(A,G)$ is realisable if and only if the \ltl formula $(A \to G)$ is realisable. Conversely, ordinary \ltl realisability can be viewed as a special case of contract realisability, since an \ltl formula $\varphi$ is realisable if and only if the contract $(\cT,\varphi)$ is realisable.

\subsection*{{\bf A sound algorithm for decomposing LTL formulae}} 

\vspace{2mm}

In the context of Assume Guarantee contracts, the work presented in \citep{I18, I24} proposes an algorithm for decomposing complex specifications expressed as \ltl formulae into smaller subspecifications. We refer to this algorithm as $\mathcal{DC}$.  

Their method uses a model checker to generate traces from which sets of independent variables are identified. These sets then serve as the basis for partitioning the original specification into independent subspecifications. The authors state that the method is sound but not complete.

In this section, we present $\mathcal{DC}$ independently of the $(A,G)$ contract framework and apply it to arbitrary \ltl formulae. We introduce the additional precondition that the input \ltl formula must be satisfiable, since an unsatisfiable formula is trivially unrealisable. This precondition can be verified through a single model-checking query that determines whether the formula is satisfiable.

We begin with two notions used in~\citep{I18, I24}. Although these concepts were originally formulated differently, our presentation captures the same underlying ideas and is formally equivalent to the original definitions.

    \begin{definition}  [Projection formula]
    \label{def:not}
    Let $\varphi(\E, \Y)$ be an \ltl formula. Let 
    $V  =  \{v_1, \ldots, v_{n_1}\} \subseteq \Y$ while 
    $\overline{V} = (\Y\setminus V) = \{w_1, \ldots, $ $w_{n_2}\}$. 
    The  projection formula of~$\varphi$ over~$V$, denoted as  $\varphi_V$, is the formula  $\varphi(\E, V, w'_1, \cdots w'_{n_2})$, where $w'_j~(1\leq j \leq n_2)$ is a fresh variable.
    \end{definition}
This notion of projection formula corresponds to the notion of projection in~\citep{I24} (see Definition 5.1).

\begin{example}
Let  $\varphi = \alw(p \to (a \vee (b \wedge c))) \wedge \ev(p \to  d) \wedge \ev(\neg p \to \neg d)$,  where $\E = \{p\}$ and $\Y = \{a, b, c, d\}$.
Hence, picking new variables $a', b', c', d'$:
\begin{eqnarray*}
\qquad \varphi_{\{a, b\}} & = & \alw(p \to (a \vee (b \wedge c'))) \wedge \ev(p \to d') \wedge \ev(\neg p \to \neg d')\\ 
\qquad \varphi_{\{c, d\}} & = &  \alw(p \to (a' \vee (b' \wedge c))) \wedge \ev(p \to d) \wedge \ev(\neg p \to \neg d).
\end{eqnarray*}
\end{example}

Next, we introduce the concept of an independent set of system variables. It corresponds to Definition 5.3 in \citep{I24}.
    
\begin{definition} [Independent set of variables]
\label{def:first_ind}
Let $\varphi(\E, \Y)$  be an \ltl formula and let $V \subseteq \Y$.
 The set $V$ is independent  in~$\varphi$ if and only if $(\varphi_V \wedge \varphi_{\overline{V}}) \rightarrow \varphi$ is a valid formula, i.e, for every trace $\sigma$, if $\sigma \models (\varphi_V \wedge \varphi_{\overline{V}})$ then  $\sigma \models \varphi$.
\end{definition}
\begin{example}
Let $\E = \{p\}$ and $\Y = \{a,b\}$. Consider the following formulae:
$
 \varphi  =   \alw(p \to \nx(a \wedge b)) $ and $ \psi =   \ev(p \to \nx(a \wedge b)). 
$
It is easy to see that every trace that is a model of 
$$\varphi_{\{a\}} \wedge \varphi_{\{b\}} =  \alw(p \to \nx(a \wedge b')) \wedge  \alw(p \to \nx(a' \wedge b))$$
also satisfies  $\varphi = \alw(p \to \nx(a \wedge b))$. Hence, the singleton sets $\{a\}$ and $\{b\}$ are independent in $\varphi$.
By contrast, this is not the case for $\psi$. Consider the trace $\sigma = \{p\}, \{p, a', b\}, \{p, a, b'\}, \{p, a\}^\omega$. Then $\sigma$ is a model of
$$\psi_{\{a\}} \wedge \psi_{\{b\}} = \ev(p \to \nx(a \wedge b')) \wedge  \ev(p \to \nx(a' \wedge b))$$
but it is not a model of $\psi =  \ev(p \to \nx(a \wedge b))$. Therefore, neither $\{a\}$ nor $\{b\}$ is independent in $\psi$.

The situation for $\varphi$ reflects the fact that $\varphi$ can be decomposed into the two smaller formulae
$\varphi_1 = \alw (p \to \nx \,a)$ and  $\varphi_2 = \alw (p \to \nx\, b)$, so that $\varphi$  and $(\varphi_1 \wedge \varphi_2)$ are logically equivalent.  Consequently, $\varphi$ is realisable if and only if both $\varphi_1$ and $\varphi_2$ are realisable.
For $\psi$, no analogous decomposition into
$
\psi_1 = \ev(p \to \nx a)$ and
$\psi_2 = \ev(p \to \nx b)
$
preserves logical equivalence with $\psi$. Intuitively, the two projected eventualities may be satisfied at different time points, while $\psi$ requires a single time point at which both $a$ and $b$ hold in the next state.
\end{example}

The following lemma establishes a syntactic property of the projection formulae used in Algorithm~${\cal DC}$.

\begin{lemma}
\label{lem:alg-first}
Let $\varphi$ be an \ltl formula and let $V \subseteq \Y$.  If there exists  a trace~$\sigma$ such that  $\sigma \models (\varphi_V \wedge \varphi_{\overline{V}} \wedge \neg \varphi)$, then there exist a state $i \geq 0$, and a variable $w \in \overline V$ such that  $(w' \in \sigma_i \iff w \not \in \sigma_i)$.
\end{lemma}

\begin{proof}
Assume that
$
\sigma \models (\varphi_V \wedge \varphi_{\overline V} \wedge \neg \varphi).
$
Suppose, towards a contradiction, that for every $i \geq 0$ and every $w \in \overline V$,
$
w' \in \sigma_i \iff w \in \sigma_i.
$
Then every fresh variable $w'$ has exactly the same truth value as $w$ at every state of~$\sigma$. Since~$\varphi_V$ is obtained from $\varphi$ by replacing each variable $w \in \overline V$ with its fresh copy $w'$, it follows that
$
\sigma \models \varphi_V  \Rightarrow  \sigma \models \varphi.
$
This contradicts the assumption that $\sigma \models \neg \varphi$.
\end{proof}

\begin{algorithm}[H]
\label{alg:one}
{\small
\caption{{\small ${\cal DC}$}}
\SetKwInOut{Input}{Input} 
\SetKwInOut{Output}{Output} 
\Input {A satisfiable \ltl formula $\varphi$ over $\E \cup\Y$}
\Output {clusters is a partition of $\Y$ into sets of independent variables}
clusters,  R $\gets \emptyset, \Y$;\\
\While {$R \neq \emptyset$}  {
choose $x \in R$;\\ 
$V$, passed  $\gets \{x\}$,  $false$;\\ 
\Repeat {passed}  {
$\overline{V} \gets \Y\setminus V$;\\ 
passed, $\sigma$ $\gets$ {\tt CheckValidity}$( (\varphi_V \wedge \varphi_{\overline{V}}) \rightarrow \varphi)$;\\ 
\If{not passed }{
$ D \gets$ {\tt ParseTrace}($\sigma$);\\
$ V  = V \cup D$;\\
}
}
clusters $\gets$ clusters $\cup \, \{V\}$;\\
$ R \gets \Y\setminus (\textstyle{\bigcup_{S \,\in\, \mbox{clusters}} \textstyle{S}} )$\ 
}
\Return clusters;
}
\end{algorithm}

\vspace{2mm}

 Algorithm~\ref{alg:one} presents the function ${\cal DC}$, which uses Definition~\ref{def:first_ind} as an operational criterion for searching independent sets of variables with the help of a model checker. Starting from a single variable, it incrementally enlarges a candidate set until the validity condition from Definition~\ref{def:first_ind} is satisfied. When the condition fails, a counterexample trace returned by the model checker is analysed in order to identify additional dependent variables that must be added to the candidate set. 
We illustrate how the function ${\cal DC}$ works by means of two examples.

\begin{example}
\label{ex:new1}
Let $\E = \{p\}$, $\Y = \{a,b\}$, and
$
\varphi = \ev(p \to \nx(a \wedge b)).
$
Suppose that ${\cal DC}$ begins by selecting $a \in \Y$ in line 3. Then, after lines 4 and 6, we have $V = \{a\}$ and $\overline{V} = \{b\}$. In line 7, the model checker is asked whether the formula
$
\Phi = (\varphi_{\{a\}} \wedge \varphi_{\{b\}}) \rightarrow \varphi
$
is valid. In this case,
$$
\Phi =
[\ev(p \to \nx(a \wedge b')) \wedge \ev(p \to \nx(a' \wedge b))]
\rightarrow
\ev(p \to \nx(a \wedge b)).
$$
Since $\Phi$ is not valid, the model checker returns a trace $\sigma$ satisfying $\neg \Phi$. For example, it may return
$
\sigma = \{p\}, \{p,a',b\}, \{p,a,b'\}, \{p,a\}^\omega.
$
Thus, after line~7, we have passed = $false$ together with the counterexample trace $\sigma$.

In line 9, the function {\tt ParseTrace} identifies the variables responsible for the failure of validity. These are precisely the variables whose values differ from those of their corresponding primed versions at some state of the trace~$\sigma$. Note that Lemma~\ref{lem:alg-first} guarantees the existence of at least one such variable. In this example, $b$ is identified in this way, so line 10 updates $V$ to $\{a,b\}$.

During the second iteration, the model checker is then asked to verify
$(\varphi_{\{a,b\}} \wedge \varphi_{\emptyset}) \rightarrow \varphi$ = 
$[\ev(p \to \nx(a \wedge b)) \wedge \ev(p \to \nx(a' \wedge b'))]
\rightarrow
\ev(p \to \nx(a \wedge b))$.
This formula is valid. Hence, passed is set to $true$  and ${\cal DC}$ terminates with the single set of independent variables $\{a,b\}$.
\end{example}

\begin{example}
\label{ex:intro}
Let $\E = \{p\}$, $\Y = \{t,v,w,x\}$, and let $\varphi$ be the formula

\vspace{-4mm}

$$ 
\alw ( (p \to \nx (t \vee v)) \wedge (\neg p  \to \nx (w \vee x) )) 
$$

\noindent Suppose that ${\cal DC}$ starts by selecting $w \in \Y$. Then it checks whether
$
\Phi = (\varphi_{\{w\}} \wedge \varphi_{\{t,v,x\}}) \rightarrow \varphi
$
is valid. In this case, $\Phi$ is not valid and the model checker may return the trace
$
\sigma = \{t', v , w', x'\}^\omega.
$

It is easy to see that $\sigma \models (\varphi_{\{w\}} \wedge \varphi_{\{t,v,x\}})$, but $\sigma \not \models \varphi$.
The failure occurs already at the initial state: since $p$ is false at state $0$, the formula $\varphi$ requires $(w \vee x)$ to hold at state $1$, but both $w$ and $x$
are false there. Therefore, $\Phi$ is invalid.

The function {\tt ParseTrace} then identifies the variables responsible for this failure. In this example, it returns
$
D = \{t,v,x\}.
$ 
Consequently, ${\cal DC}$ merges all system variables into a single cluster and outputs $\{t,v,w,x\}$.
\end{example}

The running time of {\small ${\cal DC}$} is polynomial in the number of variables in $\Y$, multiplied by the cost of the validity checks delegated to the model checker. For full \ltl, validity checking is PSPACE-complete, since an \ltl formula is valid if and only if its negation is unsatisfiable..

Furthermore, the algorithm ${\cal DC}$ is sound, because every set $V$ placed in clusters satisfies the condition of Definition~\ref{def:first_ind}, and is therefore independent in the input formula. However, ${\cal DC}$ is not complete. Consider the formula $\varphi$ in Example~\ref{ex:intro}. The algorithm may return the single set $\{t,v,w,x\}$, even though a finer decomposition exists, namely into the two independent subsets $\{t, v\}$ and $\{w,x\}$.
Indeed, the \ltl formula $\varphi$ 
is logically equivalent to the \ltl formula
$ \alw (p \to \nx (t \vee v)) \wedge  \alw (\neg p \to \nx (w \vee x))$.
Hence, the specification can be decomposed into two smaller independent parts, although ${\cal DC}$ may fail to detect this decomposition.

\section{Semantic notions underlying algorithm ${\cal DC}$} 

We formalise independence at the level of traces, using projection and join operators inspired by relational algebra~\citep{C70}.

From now on, we omit the qualifier \ltl\ and simply write formulae. Unless stated otherwise, all formulae are over variables in $\E \cup \Y$, and traces belong to $(2^{\E \cup \Y})^\omega$.
In contexts involving projection formulae, traces may also range over the expanded alphabet $\E \cup \Y \cup \Y'$, where $\Y'=\{y' \mid y \in \Y\}$ is a fresh set of primed copies of the variables in $\Y$. We use $\alpha$, $\beta$, $\gamma$ for traces over $\E \cup \Y$, and $\sigma$, $\tau$ for traces over $\E \cup \Y \cup \Y'$.

We begin with the projection operator on traces, which always preserves environment variables and keeps only the variables selected from~$\Y$~$\cup$~$\Y'$.
\begin{definition}[Trace projection]
\label{def:p-res}
 Let $\sigma$ be a trace, and let $V \subseteq \Y\cup \Y'$. The projection of $\sigma$ over $V$, denoted by $\sigma\res V$, is  the trace such that for all $i \geq 0$, $(\sigma\res V)_i = (\sigma_i \cap \E) \cup (\sigma_i \cap V)$.\\[-3mm]
\end{definition}

The next example illustrates the effect of projection on traces.
\begin{example}
\label{ex:numberone}
Consider the traces $\sigma = \{p,a,d, d'\}^\omega$ and  $\beta = \{p,b, c, d\}^\omega$. Then  $\sigma\res \{a\} = \{p,a\}^\omega$ and  $\beta\res \{a\} = \{p\}^\omega$.
\end{example}

We extend Definition \ref{def:p-res} to sets of traces. 

\begin{definition}[Projection of a set of traces]
\label{def:model_p}
Let $V \subseteq \Y \cup \Y'$ and let $\Sigma$ be a set of traces. The projection of $\Sigma$ over $V$, denoted by $\Sigma \res V$, is defined by
$
\Sigma \res V = \{\sigma \res V \mid \sigma \in \Sigma\}.
$
\end{definition}

We now introduce the join operator, again from relational algebra.

\begin{definition}[Join]
\label{def:model-c}
Let $V_1,V_2 \subseteq \Y$. Let $\Sigma_1$ and $\Sigma_2$ be sets of traces over $\E \cup V_1$ and $\E \cup V_2$, respectively. The \emph{join} of $\Sigma_1$ and $\Sigma_2$, denoted by $\Sigma_1 \Join \Sigma_2$, is the set of traces over $\E \cup V_1 \cup V_2$ defined by
$$
\Sigma_1 \Join \Sigma_2
=
\{\sigma \mid (\sigma \res V_1) \in \Sigma_1 \text{ and } (\sigma \res V_2) \in \Sigma_2\}.
$$
\end{definition}

\begin{example}
\label{ex:numbertwo}
Let
$
\varphi = \alw\bigl((p \to (a \vee b)) \wedge (\neg p \to (\neg a \wedge b)) \wedge c\bigr),
$ where $\E = \{p\}$ and  $\Y = \{a, b, c\}$.
We show that the infinite trace 
$
\{p,c\}^\omega$ belongs to $(\Sigma_{\varphi} \res \{b,c\}) \Join (\Sigma_{\varphi} \res \{a\}).
$
Consider the two traces
$
\alpha = \{p,a,c\}^\omega$
 and
 $\beta = \{p,b,c\}^\omega.
$
Both belong to $\Sigma_{\varphi}$, since both are models of $\varphi$. 
Now consider the trace
$
\gamma = \{p,c\}^\omega.
$
Then
$
\gamma \res \{b,c\} = \{p,c\}^\omega = \alpha \res \{b,c\},
$
and since $\alpha \in \Sigma_\varphi$, it follows that
$
\gamma \res \{b,c\} \in \Sigma_\varphi \res \{b,c\}.
$
Similarly,
$
\gamma \res \{a\} = \{p\}^\omega = \beta \res \{a\},
$
and since $\beta \in \Sigma_\varphi$, we have
$
\gamma \res \{a\} \in \Sigma_\varphi \res \{a\}.
$
Therefore, by Definition~\ref{def:model-c} of join,
$
\gamma = \{p,c\}^\omega \in (\Sigma_{\varphi} \res \{b,c\}) \Join (\Sigma_{\varphi} \res \{a\}). 
$
Note that $\gamma$ does not belong to~$\Sigma_\varphi$.
\end{example}

The following properties of the projection and join operators are standard in relational algebra~\citep{YP82} and remain valid in the present setting.

 Let $U, V, V_1, V_2 \subseteq \Y$ and
let $\Sigma$, $\Sigma_1$, $\Sigma_2$  be three sets of traces over $\E \cup \Y$, $\E \cup V_1$, and $\E \cup V_2$, respectively. Then the following properties hold: 

\begin{enumerate}
\item[$(a)$]
Commutativity,  associativity, and monotonicity.
\begin{itemize}
\item[--]
$\Sigma_1  \Join \Sigma_2 = \Sigma_2 \Join \Sigma_1 $
\item[--]
$(\Sigma_1  \Join \Sigma_2) \Join \Sigma = \Sigma_1 \Join (\Sigma_2 \Join \Sigma) $
\item[--]
If $\Sigma_1 \subseteq \Sigma'_1$, then  $\Sigma_1 \Join \Sigma \subseteq \Sigma'_1 \Join \Sigma$.\\[-3mm]
\end{itemize}
\item[$(b)$]
Idempotency of join. 
\begin{itemize}
\item[--]
$\Sigma  \Join (\Sigma\res U)  = \Sigma$. In particular $\Sigma  \Join (\Sigma\res \emptyset)  = \Sigma$ and 

$(\Sigma\res U)  \Join (\Sigma\res U)  = (\Sigma\res U) $. \\[-3mm]
\end{itemize}
\item[$(c)$]
Compatibility of projection with join. 
\begin{itemize}
\item[--]
$\Sigma\res (U \cup V) \subseteq (\Sigma\res U) \Join (\Sigma\res V)$\\[-3mm]
\end{itemize}
\item[$(d)$]
Distributivity of join over projection. \\
Let $X \subseteq V_1$ and $Y \subseteq V_2$ with $V_1 \cap V_2  = \emptyset$. 
\begin{itemize}
\item[--]
$(\Sigma_1 \Join \Sigma_2)\res (X \cup Y) = (\Sigma_1\res X) \Join (\Sigma_2\res Y)$
\end{itemize}
\end{enumerate}

\section{The meaning of independence in the ${\cal DC}$ algorithm}

In this section, we formalise the notion of independence underlying ${\cal DC}$ in terms of the join operator.

\begin{definition}[Independence via the join operator]
\label{def:model-ind}
Let $\varphi$ be a formula, and let $V \subseteq \Y$. We say that $V$ is \emph{join} independent in $\varphi$ if and only if
$$
(\Sigma_{\varphi} \res V) \Join (\Sigma_{\varphi} \res \overline{V}) = \Sigma_{\varphi},
$$
where $\overline{V} = \Y \setminus V$.
\end{definition}

The following lemma reveals the close relationship between models of projection formulae and projections of models of these formulae.

\begin{lemma}
\label{lem:totallynew}
Let $\varphi$ be a formula, and let $V, W \subseteq \Y$ be such that $V \cup W = \Y$. Let $Z = \{z_1, z_2, \cdots z_k\} \subseteq \Y$.
\begin{itemize}
\item[(i)]
Let $\sigma$ be a trace over the expanded alphabet containing both the variables in $\Y$ and their primed copies. If
$
\sigma \models (\varphi_V \wedge \varphi_W \wedge \bigwedge_{i=1}^k \alw(z_i \leftrightarrow z_i' )),
$
then the projected trace
$
\sigma \res \Y \in (\Sigma_\varphi \res (V \cup Z)) \Join (\Sigma_\varphi \res (W\cup Z)).
$
\item[(ii)]
Let $\alpha$ be a trace over $\E \cup \Y$. If 
$
\alpha \in (\Sigma_\varphi \res (V \cup Z)) \Join (\Sigma_\varphi \res (W\cup Z)),
$
then there exists an extension of $\alpha$, denoted by $\widehat\alpha$, over the expanded alphabet such that
$
\widehat\alpha \models (\varphi_V \wedge \varphi_W \wedge \bigwedge_{i=1}^k \alw(z_i \leftrightarrow z_i' )).
$
\end{itemize}
\end{lemma}

\begin{proof}
For item $(i)$, assume that
$
\sigma \models (\varphi_V \wedge \varphi_W \wedge \bigwedge_{i=1}^k \alw(z_i \leftrightarrow z_i' ))$.
From $\sigma \models \varphi_V$, we define a trace $\gamma$ over $\E \cup \Y$ as follows. The values of the variables in $V$ are kept unchanged, while each variable in $W$ is interpreted according to the value of its primed counterpart in $\sigma$. Note that $\varphi_V$ does not contain primed copies of variables in $V \cap W$. Moreover,  $\varphi_V  \wedge \bigwedge_{i=1}^k \alw(z_i \leftrightarrow z_i' )$ enforces that the variables in~$Z$ take the same value in both their primed and non-primed versions.
 Then $\gamma \models \varphi$, and
$
\gamma \res (V \cup Z) = \sigma \res (V \cup Z).
$
Hence
$
\sigma \res (V \cup Z) \in \Sigma_\varphi \res (V \cup Z).
$

Similarly, from $\sigma \models \varphi_W$, we obtain $
\sigma \res (W \cup Z) \in \Sigma_\varphi \res (W \cup Z)$.
Therefore, by Definition~\ref{def:model-c},
$
\sigma \res \Y \in (\Sigma_\varphi \res (V \cup Z) \Join (\Sigma_\varphi \res (W \cup Z)).
$

\vspace{2mm}

For item $(ii)$, assume that
$\alpha \in (\Sigma_{\varphi}\res (V \cup Z))   \Join  (\Sigma_{\varphi}\res (W \cup Z))$.
Since 
$\alpha \res (V \cup Z)$ belongs to $(\Sigma_{\varphi} \res (V \cup Z))$, there exists an extension of  $\alpha \res (V \cup Z)$, denoted by~$\alpha_v$, which is a model of~$\varphi$. Similarly,  since 
$\alpha \res (W \cup Z)$ belongs to $(\Sigma_{\varphi} \res (W \cup Z))$, there exists an extension of  $\alpha\res (W \cup Z)$, denoted by $ \alpha_w$, which is a model of~$\varphi$. 
We define the extension of  $\alpha$, denoted by $\widehat{\alpha}$, as the trace that, for every state $i \geq 0$:
\begin{eqnarray*}
\widehat{\alpha}_i & = & \alpha_i \cup  \{z' \mid z \in Z \mbox{ and } z \in \alpha_i\}  \cup \\
& &  \{v' \mid v \in V\setminus Z \mbox{ and } v \in (\alpha_w)_i \}  \cup \\
& & \{w' \mid w \in W\setminus Z \mbox{ and } w \in (\alpha_v)_i \}.
\end{eqnarray*}
By construction,  $\widehat{\alpha}\models (\varphi_V \wedge \varphi_W \wedge \bigwedge_{i=1}^k \alw(z_i \leftrightarrow z_i' ))$. 

\end{proof}

\begin{example}
\label{ex:surprise0}
Let $\E = \{p\}$, $\Y = \{a,b,c\}$, and
$$
\varphi = \ev((p \to (a \vee b)) \wedge c) \wedge (\neg p \to \neg c).
$$
Consider the sets $V= \{a,b\}$, $W = \{a,c\}$, $Z = \emptyset$, and  the trace
$
\alpha = \{p,a\}, \{p,c\}^\omega.
$
We show that
$$
\alpha \in (\Sigma_\varphi \res \{a,b\}) \Join (\Sigma_\varphi \res \{a, c\}).
$$
Indeed,
$
\alpha \res \{a,b\} = \{p,a\}, \{p\}^\omega,
$
and this belongs to $\Sigma_\varphi \res \{a,b\}$, since the extension
$
\alpha_{\{a,b\}} = \{p,a,c\}, \{p\}^\omega
$
is a model of $\varphi$. Also,
$
\alpha \res\{a, c\} = \alpha,
$
and this belongs to $\Sigma_\varphi \res \{a, c\}$, since the extension
$
\alpha_{\{a,c\}} = \{p, a\}, \{p,b,c\}^\omega
$
is  a model of $\varphi$.
Hence $\alpha$ belongs to
$
(\Sigma_\varphi \res \{a,b\}) \Join (\Sigma_\varphi \res \{a, c\}).
$

Moreover, $\alpha$ can be extended to the trace $\widehat{\alpha} = \{p,a,a',c'\}, \{p,c,b'\}^\omega$, which is a model of 
$
\varphi_{\{a,b\}} \wedge \varphi_{\{a, c\}}. 
$
\end{example}

The next lemma establishes the equivalence between the notion of independent sets as defined in~\citep{I18, I24}   and the notion of join-independent sets introduced in Definition~\ref{def:model-ind}.

\begin{lemma}
\label{lem:first}
Let $\varphi$ be a  formula, and let $V \subseteq \Y$. Then $V$ is join-independent in $\varphi$ if and only if the formula
$
(\varphi_V \wedge \varphi_{\overline V}) \rightarrow \varphi
$
is valid. That is, for every trace $\sigma$, if
$
\sigma \models (\varphi_V \wedge \varphi_{\overline V}),
$
then
$
\sigma \models \varphi.
$
\end{lemma}

\begin{proof}
By Definition~\ref{def:model-ind}, the set $V$ is join-independent in~$\varphi$ if and only if
\begin{equation}
\label{eq:one}
(\Sigma_{\varphi}\res V) \Join (\Sigma_{\varphi}\res \overline{V}) = \Sigma_{\varphi}.
\end{equation}
Assume first that Equation~\eqref{eq:one} holds. Let $\sigma$ be a trace such that it is a model of $(\varphi_V \wedge \varphi_{\overline{V}})
$.  By Lemma~\ref{lem:totallynew}, item $(i)$, taking as $W$ the set $\overline V$ and $Z = \emptyset$, we have
$
 \sigma\res \Y \in  
  (\Sigma_{\varphi}\res V) \Join (\Sigma_{\varphi}\res \overline{V})   
$.
Hence $\sigma \res \Y \models \varphi$. Since $\varphi$ contains only variables from $\E \cup \Y$, its truth depends only on the restriction of the trace to these variables. Therefore,
$
\sigma \models \varphi.
$
This proves that
$
(\varphi_V \wedge \varphi_{\overline{V}}) \rightarrow \varphi
$
is valid.

Conversely, assume that for every trace $\sigma$ that is a model  of $(\varphi_V \wedge  \varphi_{\overline{V}})$, we have $\sigma \models \varphi$.
Take any trace~$\alpha$  in  $(\Sigma_{\varphi}\res V) \Join  (\Sigma_{\varphi}\res \overline{V})
$. 
 By Lemma~\ref{lem:totallynew}, item $(ii)$, (taking $W = \overline V$ and $Z = \emptyset$), there exists an extension $\widehat{\alpha}$  of $\alpha$ such that  $\widehat{\alpha} \models (\varphi_V \wedge  \varphi_{\overline{V}})$. By assumption, it follows that $\widehat{\alpha}\models \varphi$. Again, since~$\varphi$ depends only on the original variables, we obtain  $
\widehat{\alpha} \res \Y = \alpha$ and $\alpha \models \varphi$. Hence $
\alpha \in \Sigma_{\varphi}.$
 Therefore, $(\Sigma_{\varphi}\res V) \Join  (\Sigma_{\varphi}\res\ \overline{V}) \subseteq\Sigma_{\varphi}$. 
  For the reverse inclusion, since 
 every trace in $\Sigma_\varphi$ is over $\E \cup \Y$, we have
$
\Sigma_\varphi \res \Y = \Sigma_\varphi.
$
 As $V \cup \overline V = \Y$, property~(c) yields
$$
\Sigma_\varphi
=
\Sigma_\varphi \res \Y
=
\Sigma_\varphi \res (V \cup \overline V)
\subseteq
(\Sigma_\varphi \res V) \Join (\Sigma_\varphi \res \overline V).
$$
 So $(\Sigma_{\varphi}\res V) \Join (\Sigma_{\varphi}\res \overline{V}) = \Sigma_{\varphi}$ and $V$ is join-independent in~$\varphi$.
 \end{proof}

We will use the contrapositive of Lemma~\ref{lem:first}, stated below. 

\begin{corollary}
\label{cor:first}
Let $\varphi$ be a formula, and let $V \subseteq \Y$. The set $V$ is not join-independent in $\varphi$ if and only if  $(\varphi_V \wedge \varphi_{\overline V} \wedge \neg \varphi)$ is satisfiable.
\end{corollary}

The previous lemma and corollary establish that the notion of independence in Definition~\ref{def:first_ind} is equivalent to the notion of join-independence. From now on, we simply speak of independent and dependent sets of variables, where dependent means not independent. In addition, for simplicity, we will sometimes omit  \emph{(in)dependent in} $\varphi$ and write only \emph{(in)dependent} whenever~$\varphi$ is clear from the context.

\begin{example}
\label{ex:surprise}
Let $\E = \{p\}$, $\Y = \{a, b, c\}$, and 
$$\varphi = \ev((p \to (a \vee b)) \wedge c) \wedge (\neg p \to  \neg c)$$ 
as in Example~\ref{ex:surprise0}.
Consider the trace $\sigma = \{p, a, c'\}, \{p, a', c\}, \{p, c\}^\omega$. Clearly, $\sigma \models \varphi_{\{a,b\}} \wedge \varphi_{\{c\}}$, where
\begin{eqnarray*}
\varphi_{\{a,b\}}  & = &    \ev((p \to (a \vee b)) \wedge c') \wedge (\neg p \to  \neg c')  \\
\varphi_{\{c\}} & = &  \ev((p \to (a' \vee b')) \wedge c) \wedge (\neg p \to  \neg c). 
 \end{eqnarray*}
However, $\sigma \not\models \varphi$. Indeed, the eventuality in $\varphi$ requires some time point at which both $c$ and $(a \vee b)$ hold. In $\sigma$, $a$ holds only at the first state, where $c$ is false, while $c$ holds from the second state onward, where neither $a$ nor $b$ is true. Hence no state satisfies $(p \to (a \vee b)) \wedge c$. Therefore, by Corollary~\ref{cor:first}, both sets  $\{a,b\}$ and $\{c\}$ are dependent in $\varphi$.

Moreover,
$
\sigma \res \Y = \{p,a\}, \{p,c\}^\omega
$
belongs to
$
(\Sigma_\varphi \res \{a,b\}) \Join (\Sigma_\varphi \res \{c\}).
$
Nevertheless,
it does not belong to $\Sigma_\varphi$. Thus,
$
(\Sigma_\varphi \res \{a,b\}) \Join (\Sigma_\varphi \res \{c\}) \not\subseteq \Sigma_\varphi.
$
\end{example}

The next lemma gives another characterisation of independent sets.

\begin{lemma}
\label{lem:cero}
Let $\varphi$ be a formula, and let $V \subseteq \Y$. Then $V$ is independent in~$\varphi$ if and only if for every set $U \subseteq \overline V$,
$
\Sigma_\varphi \res (U \cup V) = (\Sigma_\varphi \res U) \Join (\Sigma_\varphi \res V).
$
\end{lemma}

\begin{proof}
For the forward direction, assume that $V$ is independent. By Definition~\ref{def:model-ind}.

\vspace{-0.7cm}

\begin{eqnarray}
\label{eq:new}
\hspace{-3cm} \Sigma_{\varphi} = (\Sigma_{\varphi}\res V)   \Join  (\Sigma_{\varphi}\res \overline{V})
 \end{eqnarray}

\noindent
 Projecting in the equation~(\ref{eq:new}) over $V \cup U$, we obtain:
\begin{eqnarray*}  
&&\Sigma_{\varphi} \res (V \cup U) \; =   [(\Sigma_{\varphi}\res V)   \Join  (\Sigma_{\varphi}\res \overline{V})]\res (V \cup U) \; =\\
& & \mbox{\scriptsize{ by distributivity (see $(d)$) with $\Sigma_1 = (\Sigma_{\varphi}\res V), \Sigma_2 =(\Sigma_{\varphi}\res \overline{V}), X = V, \mbox{ and } Y = U$}}  \\
&& (\Sigma_{\varphi}\res V)   \Join  ((\Sigma_{\varphi} \res \overline{V})\res U) = \\
&& (\Sigma_{\varphi}\res V)   \Join  (\Sigma_{\varphi}\res U)  = \\
&&  \mbox{\scriptsize{ by commutativity (see $(a)$)}} \\
&&(\Sigma_{\varphi}\res U)   \Join  (\Sigma_{\varphi}\res V). 
\end{eqnarray*}
Thus, $\Sigma_{\varphi} \res (V \cup U) = (\Sigma_{\varphi}\res U)   \Join  (\Sigma_{\varphi}\res V)$.\\
For the backward direction, assume that for every $U \subseteq \overline V$,
$$
\hspace{-3cm} \Sigma_\varphi \res (U \cup V) = (\Sigma_\varphi \res U) \Join (\Sigma_\varphi \res V).
$$
Taking $U = \overline V$, we get
$
\Sigma_\varphi
=
\Sigma_\varphi \res (\overline V \cup V)
=
(\Sigma_\varphi \res \overline V) \Join (\Sigma_\varphi \res V).
$
Hence, by Definition~\ref{def:model-ind}, the set $V$ is independent in $\varphi$.

\end{proof}

With a clearer understanding of what it means for a set to be independent, the next section presents properties that enable us to conclude with a more sophisticated version of the ${\cal DC}$ algorithm that achieves completeness.

\section{Properties to ensure completeness}

To ensure completeness, it is not sufficient to detect that a set of variables is dependent. What is really needed is a criterion for identifying the variables that, although outside a given set $V$, must necessarily be included in it, to ensure that non-empty proper subsets of $V$ cannot be independent.
To establish this, one could use the  characterisation provided by Lemma~\ref{lem:cero}. However, this approach is highly inefficient since proving that a set $V$ is dependent would require finding a subset $U \subseteq \overline V$ such that $
(\Sigma_\varphi \res V) \Join (\Sigma_\varphi \res U) \neq \Sigma_\varphi \res (V \cup U)
$. On the other hand, when deciding whether a new variable $z$ must be included in a dependent set $V$, it is not sufficient to rely solely on the simple join test
$$
(\Sigma_\varphi \res V) \Join (\Sigma_\varphi \res \{z\}) \neq \Sigma_\varphi \res (V \cup \{z\}).
$$
By Lemma~\ref{lem:cero}, this test provides a sufficient condition to conclude that the dependency of $V$ is already induced by the variable $z$, and therefore that $z$ must be included in $V$. However, the converse does not hold. A variable may be forced to belong to a dependent $V$ even though this is not detected by the above test. Indeed, the necessity of such a variable may become visible only after fixing the behaviour of some additional variables. The next example illustrates this situation.

\begin{example}
Let $\varphi = \alw( p \to (a \vee (b \wedge c)))$, where $p \in \E$ and $\Y = \{a, b, c\} $. Using Lemma~\ref{lem:cero}, we observe that  $\{b\}$ is 
 dependent, since 
 $$(\Sigma_{\varphi}\res \{b\})   \Join  (\Sigma_{\varphi}\res \{a\}) \neq (\Sigma_{\varphi}\res \{a,b\}).$$
 Likewise, the set $\{a,b\}$ is dependent, because $(\Sigma_{\varphi}\res \{a, b\})   \Join  (\Sigma_{\varphi}\res \{c\}) \neq \Sigma_{\varphi}$.  Consequently, apart from the empty set, the only independent set in $\varphi$ is $\{a, b, c\}$.
 What is important here is that the relationship between the variables $b$ and $c$ is mediated by the variable~$a$. However, the direct test between $\{b\}$ and $\{c\}$ does not reveal this, since
$
(\Sigma_\varphi \res \{b\}) \Join (\Sigma_\varphi \res \{c\}) = \Sigma_\varphi \res \{b,c\}.
$
\end{example}

This example motivates the search for a more refined yet efficient mechanism that, given that $\{b\}$ is dependent, guarantees that the next variable selected to extend $\{b\}$ is $a$ rather than $c$.

Before presenting this refined mechanism, we establish two additional properties that are crucial for proving the completeness of the method.

\begin{lemma}
\label{lem:complement}
Let $\varphi$ be a formula and let $V \subseteq \Y$ be independent in $\varphi$.
If $U \subsetneq V$ is an independent subset, then $V \setminus U$ is
also independent.
\end{lemma}

\begin{proof}
To show that $V \setminus U$ is
independent, using Definition~\ref{def:model-ind}, we prove that 
$(\Sigma_\varphi\res V\setminus U) \Join (\Sigma_\varphi \res \overline{V\setminus U}) = \Sigma_\varphi$.
Since $U$ is independent, by Lemma~\ref{lem:cero}, 
\begin{equation}
\label{eq:otra1}
(\Sigma_\varphi \res U) \Join (\Sigma_\varphi\res V\setminus U) = (\Sigma_\varphi\res V)    
\end{equation}
Furthermore, since $V$ is independent,
\begin{equation}
\label{eq:otra2}
(\Sigma_\varphi \res V) \Join (\Sigma_\varphi\res \overline V) = \Sigma_\varphi    
\end{equation}
Then, 

\vspace{-0.8cm}

 \begin{eqnarray*} 
&& (\Sigma_\varphi\res V\setminus U) \Join (\Sigma_\varphi \res \overline{V\setminus U})\; = \\
&& (\Sigma_\varphi\res V\setminus U) \Join (\Sigma_\varphi \res \overline V \cup U)\; \subseteq \\
& & \mbox{\scriptsize{ by the $\Join$ property $(c)$ and monotonicity $(a)$}}  \\
& & (\Sigma_\varphi\res V\setminus U) \Join [(\Sigma_\varphi \res \overline V ) \Join (\Sigma_\varphi \res U )]\; =  \\
& & \mbox{\scriptsize{ by the $\Join$ associativity and commutativity (see $(a)$) }}  \\
&&   [(\Sigma_\varphi\res V\setminus U) \Join (\Sigma_\varphi \res U )] \Join (\Sigma_\varphi \res \overline V ) \; =  \\
& & \mbox{\scriptsize{ by the $\Join$ commutativity and Equation~\ref{eq:otra1} }}  \\
& & (\Sigma_{\varphi}\res V)   \Join  (\Sigma_{\varphi}\res \overline V) = \Sigma_{\varphi}  \mbox{\scriptsize{ using Equation~\ref{eq:otra2}. }} 
\end{eqnarray*}
Hence, $(\Sigma_\varphi\res V\setminus U) \Join (\Sigma_\varphi \res \overline{V\setminus U}) \subseteq \Sigma_\varphi$. For the reverse inclusion, property $(c)$ yields
$
\Sigma_\varphi
=
\Sigma_\varphi \res (V\setminus U \cup \overline{V\setminus U})
\subseteq
(\Sigma_\varphi\res V\setminus U) \Join (\Sigma_\varphi \res \overline{V\setminus U})
$. This concludes the proof of the lemma.
\end{proof}

The next lemma shows that, once an independent superset $U$ of $V$ is fixed, variables outside $U$ cannot witness the dependence of $V$.


\begin{lemma}
\label{lem:escape}
Let $\varphi$ be a formula and let $V \subseteq W \subseteq \Y$, where $W$ is independent in $\varphi$.
Let $Z = \{z_1, z_2, \cdots, z_k\} \subseteq \overline V$.
If
$
(\varphi_V \wedge \varphi_{\overline V} \wedge \neg \varphi \wedge \bigwedge_{i=1}^k \alw(z_i \leftrightarrow z_i'))
$
is satisfiable, then for every $q \in \overline W$,
$
(\varphi_V \wedge \varphi_{\overline V} \wedge \neg \varphi \wedge \bigwedge_{i=1}^k \alw(z_i \leftrightarrow z_i') \wedge \alw(q \leftrightarrow q'))
$
is satisfiable.
\end{lemma}

\begin{proof}
Assume that
$
(\varphi_V \wedge \varphi_{\overline V} \wedge \neg \varphi \wedge \bigwedge_{i=1}^k \alw(z_i \leftrightarrow z_i'))
$
is satisfiable. Then there exists a trace $\sigma$ such that
$
\sigma \models
(\varphi_V \wedge \varphi_{\overline V} \wedge \neg \varphi \wedge \bigwedge_{i=1}^k \alw(z_i \leftrightarrow z_i')).
$
By Lemma~\ref{lem:totallynew}, item $(i)$, since $Z \subseteq \overline V$,
$
\sigma \res \Y \in (\Sigma_\varphi \res (V \cup Z)) \Join (\Sigma_\varphi \res \overline V).
$
 Given that $\sigma \models \neg \varphi$ and that $\varphi$ depends only on variables in $\E \cup \Y$, we also have
$
\sigma \res \Y \notin \Sigma_\varphi.
$
Hence
\begin{equation}
\label{eq:escape-short-1}
(\Sigma_\varphi \res (V \cup Z)) \Join (\Sigma_\varphi \res \overline V) \neq \Sigma_\varphi.
\end{equation}
Let $q \in \overline W$, and define
$
A = (V \cup Z) \cap W$ and 
$B = (V \cup Z) \cap \overline W.
$
Then $A \subseteq W$, $B \subseteq \overline W$, and $A \cup B = V \cup Z$.
Since $W$ is independent in $\varphi$,
$$
(\Sigma_\varphi \res W) \Join (\Sigma_\varphi \res \overline W)=\Sigma_\varphi.
$$
 Projecting in this equation over  $A \cup B$ and over $A \cup B \cup \{q\}$, and using the distributivity property~$(d)$, we obtain
\begin{align}
(\Sigma_\varphi \res A) \Join (\Sigma_\varphi \res B)
&=
\Sigma_\varphi \res (V \cup Z),
\label{eq:escape-short-2}
\\
(\Sigma_\varphi \res A) \Join (\Sigma_\varphi \res (B \cup \{q\}))
&=
\Sigma_\varphi \res (V \cup Z \cup \{q\}).
\label{eq:escape-short-3}
\end{align}
Now $B \subseteq \overline W \subseteq \overline V$ and $q \in \overline W \subseteq \overline V$, so
$
B \cup \{q\} \subseteq \overline V.
$
Therefore, by idempotency of join,
$$
(\Sigma_\varphi \res B) \Join (\Sigma_\varphi \res \overline V)
=
(\Sigma_\varphi \res (B \cup \{q\})) \Join (\Sigma_\varphi \res \overline V)
=
\Sigma_\varphi \res \overline V.
$$
Using \eqref{eq:escape-short-2}, \eqref{eq:escape-short-3}, associativity, and commutativity, it follows that
$$
(\Sigma_\varphi \res (V \cup Z)) \Join (\Sigma_\varphi \res \overline V)
=
(\Sigma_\varphi \res (V \cup Z \cup \{q\})) \Join (\Sigma_\varphi \res \overline V).
$$
Together with \eqref{eq:escape-short-1}, this yields
$$
(\Sigma_\varphi \res (V \cup Z \cup \{q\})) \Join (\Sigma_\varphi \res \overline V) \neq \Sigma_\varphi.
$$
Hence there exists a trace $\alpha$ such that
$$
\alpha \in (\Sigma_\varphi \res (V \cup Z \cup \{q\})) \Join (\Sigma_\varphi \res \overline V)
\quad\text{and}\quad
\alpha \notin \Sigma_\varphi.
$$
By Lemma~\ref{lem:totallynew}, item $(ii)$, since $Z \cup \{q\} \subseteq \overline V$, there exists an extension $\widehat\alpha$ of $\alpha$ such that
$
\widehat\alpha \models
(\varphi_V \wedge \varphi_{\overline V}
\wedge \bigwedge_{i=1}^k \alw(z_i \leftrightarrow z_i')
\wedge \alw(q \leftrightarrow q')).
$
Since $\alpha \notin \Sigma_\varphi$ and~$\varphi$ depends only on variables in $\E \cup \Y$, it follows that
$
\widehat\alpha \models \neg \varphi.
$
Therefore,
$
(\varphi_V \wedge \varphi_{\overline V} \wedge \neg \varphi
\wedge \bigwedge_{i=1}^k \alw(z_i \leftrightarrow z_i')
\wedge \alw(q \leftrightarrow q'))
$
is satisfiable.
\end{proof}

\section{A more sophisticated algorithm that is complete}

In this section, we propose a modification of ${\cal DC}$ in which the {\tt ParseTrace} function (Algorithm~\ref{alg:one}) is replaced by a new function, {\tt NewParseTrace} (see Algorithm~\ref{alg:two}). We then use Lemmas~\ref{lem:complement} and~\ref{lem:escape} to prove the correctness of the resulting algorithm, denoted by ${\cal NDC}$.

\vspace{2mm}

\begin{algorithm}[H]
\label{alg:two}
{\small
\caption{{\small {\tt NewParseTrace}}}
\SetKwInOut{Input}{Input} 
\SetKwInOut{Output}{Output} 
\Input{$\varphi, V,$ and a trace $\sigma \models (\varphi_V \wedge \varphi_{\overline{V}} \wedge \neg \varphi)$  provided by {\small ${\cal DC}$} }
\Output{a set $\{z\}$ where the variable witnesses the dependency of $V$} 
passed,   binds,  $\tau   \gets$ $false$,  $\cF$,   $\sigma$;\\
\While {not passed}  {
$Z =\{z \in \overline{V}:  z \in \tau_i \iff z' \not \in \tau_i  \mbox{ at some state } i \geq 0\}$;\\
$z \gets$ any variable in $Z$;\\
binds $\gets$ (binds $\vee$  $ \ev( z \leftrightarrow  \neg z'))$;\\
passed,  $\tau$ $\gets$ 
 {\tt CheckValidity}(binds $\vee$ $((\varphi_V \wedge \varphi_{\overline{V}}) \rightarrow \varphi))$;\\ 
}
\Return $\{z\}$;
}
\end{algorithm}

\vspace{2mm}

We explain ${\cal NDC}$ using Example~\ref{ex:intro}, previously introduced in Section~\ref{sec:preliminaries}. Unlike ${\cal DC}$, which identified only a single independent set, ${\cal NDC}$ successfully identifies two independent sets.

\subsection*{{\bf  Going back to Example~\ref{ex:intro}.}}

\noindent The input to the function ${\cal DC}$ was the formula
$$ 
\alw ( (p \to \nx (t \vee v)) \wedge (\neg p  \to \nx (w \vee x) )) 
$$
${\cal DC}$ started with $V = \{w\}$ and when asked about the validity of the formula  $(\varphi_{\{w\}} \wedge \varphi_{\{t, v, x\}}) \rightarrow \varphi$, the model checker returned  the  trace
 $\sigma = \{t', v , w', x'\}^\omega$.
Now, the new function {\tt NewParseTrace} is invoked with $\varphi$, $V$, and $\sigma$. In Line~1, the boolean variable passed is initialized to $false$, the formula binds is set to the constant formula $\cF$, and the trace $\tau$ is set to the input trace $\sigma$.
  
  The iteration of the function starts in Line~2, where  the set $Z$ (Line~3) is defined with all variables in $\overline{V}$, whose values  differ from those of their respective primed versions
at some state of the trace $\tau$. In this example $Z = \{t, v, x\}$.
In Line~4, some variable of $Z$ is selected and assigned to $z$. Suppose the chosen variable is $t$. Then, the function updates the variable binds with  the formula $($binds $\vee \; \ev( z \leftrightarrow  \neg z'))$ in Line~5. In the example, binds is set to $(\cF \vee \ev( t \leftrightarrow  \neg t'))$, which is equivalent to $\ev( t \leftrightarrow  \neg t')$. 

 Next, in Line~6, the function queries the model checker about the validity of the formula
$$\ev( t \leftrightarrow  \neg t') \vee ((\varphi_V \wedge \varphi_{\overline{V}}) \rightarrow \varphi)).$$
A possible response from the model checker could be the trace
$
\{t, t', v, w', x'\}^\omega
$.
At this point, since passed is $false$, another iteration begins with  $Z=\{v, x\}$.
If the function chooses $v \in Z$, the formula $(\ev( t \leftrightarrow  \neg t') \vee \ev( v \leftrightarrow  \neg v'))$ is assigned to binds, and the question of whether the formula
\[
\ev(t \leftrightarrow \neg t')
\;\vee\;
\ev(v \leftrightarrow \neg v')
\;\vee\;
\bigl((\varphi_V \wedge \varphi_{\overline{V}}) \rightarrow \varphi\bigr)
\]
is valid can be answered by the model checker with 
$
\{t,t',v,v',w',x'\}^{\omega}.
$

Since passed is $false$, another iteration begins with  $Z=\{x\}$ and the validity of the formula 
\[
\ev(t \leftrightarrow \neg t')
\;\vee\;
\ev(v \leftrightarrow \neg v')
\;\vee\;
\ev(x \leftrightarrow \neg x')
\;\vee\;
\bigl((\varphi_V \wedge \varphi_{\overline{V}}) \rightarrow \varphi\bigr)
\]
is certified by the model checker. Since
passed is now $true$, the iteration stops and the function returns~$\{x\}$ to the main algorithm, which adds  the variable~$x$ to $D$ and updates $V$ to $\{w, x\}$. 
The algorithm then queries the validity 
of the formula $(\varphi_{\{w,x\}} \wedge \varphi_{\{t,v\}} \to \varphi)$, receiving an affirmative answer. Therefore, it returns two independent sets: $\{t, v\}$ and $\{w, x\}$.\\[-2mm]

This example will be used in the last section. It is precisely Example~\ref{ex:numberzero} presented earlier. 
\begin{example}
\label{ex:last}
Let $\varphi = \alw( (\neg p \to \alw \neg b ) \wedge (p \to (\alw a \, \vee \, \alw b) ))$, where $\E = \{p\}$ and  $\Y = \{a, b\}$.
${\cal NDC}$
 may start by selecting  $a \in \Y$, and querying the model checker about the validity of $\Phi = (\varphi_{\{a\}} \wedge \varphi_{\{b\}}) \rightarrow \varphi$.
\begin{eqnarray*} 
\Phi = & ( &  \alw( (\neg p \to \alw \neg b' ) \wedge (p \to (\alw a \, \vee \, \alw b')))  \wedge\\
& & \alw( (\neg p \to \alw \neg b ) \wedge (p \to (\alw a' \, \vee \, \alw b)))\\
& ) & \to \alw( (\neg p \to \alw \neg b ) \wedge (p \to (\alw a \, \vee \, \alw b)))
\end{eqnarray*}
Since $\Phi$ is not valid,  the model checker may return  the trace $\sigma = \{p, a', b'\}^\omega$.
Then, the {\tt NewParseTrace} function is invoked with $\sigma$, $\{a\}$, and $\varphi$. With $Z = \{ b \}$, {\tt NewParseTrace} queries the model checker about the validity of 
$\ev( b \leftrightarrow  \neg b') \vee \Phi$. Since the model checker confirms that this formula is valid, the function returns $\{b\}$ and the algorithm  concludes by establishing  that $\{a, b\}$ is an independent set.
\end{example}

${\cal NDC}$  is sound and  its running
time is polynomial in~$|\Y|$ times the
cost of each query to the model-checker.
We now prove its correctness by showing that it always returns minimal independent sets---meaning that if a set of variables is classified as independent, none of its non-empty proper subsets  are independent.

\subsection*{\bf Correctness of the modified algorithm} 
 The soundness of ${\cal NDC}$ is immediate, while its completeness is established in the following theorem.

\begin{theorem}
\label{thm:look}
Given a formula $\varphi$ as input to ${\cal NDC}$, each iteration of the main loop computes a set $V$ that is independent in $\varphi$ and has no non-empty proper independent subsets.
\end{theorem}

\begin{proof}
In each iteration,  ${\cal NDC}$ takes a variable $x$ and starts with $V$~=~$\{x\}$. 
Then it executes the inner loop, adding new variables to $V$ until the formula
$
(\varphi_V \wedge \varphi_{\overline V}) \to \varphi
$
becomes valid. By Lemma~\ref{lem:first}, this means that $V$ is independent.

To show that no non-empty proper subsets of $V$ are independent, suppose to the contrary that such a subset exists.
Lemma~\ref{lem:complement} implies that
$V \setminus U$ is also independent.
Thus, possibly replacing $U$ by $V \setminus U$, we may assume that the initial
variable $x$ belongs to $U$.
Let
$
V_0=\{x\},\; V_1,\; \dots,\; V_k=V
$
be the sequence of sets constructed by the inner iteration of ${\cal NDC}$.
We show by induction that $V_i \subseteq U$ for every  $0 \leq i \leq k$.

The base case $V_0 \subseteq U$ is immediate.
Assume that $V_i \subseteq U$. 
At stage $i$, the iteration maintains a formula of the form
$$
((\varphi{_{V_i}} \wedge \varphi{_{\overline{V_i}}} )\to \varphi) \vee \bigvee_{j=1}^r \ev(z_i \leftrightarrow \neg z_i' ) 
$$
for some $\{z_1, z_2, \cdots , z_r\} \subseteq \overline{V_i}$, that is not valid. 

Suppose that the algorithm adds a variable $z \not \in U$.
By construction, this means that
$$
((\varphi{_{V_i}} \wedge \varphi{_{\overline{V_i}}} )\to \varphi) \vee \bigvee_{j=1}^r \ev(z_i \leftrightarrow \neg z_i' )  \vee \ev(z \leftrightarrow \neg z')
$$
is valid. In other words,
$$
(\varphi{_{V_i}} \wedge \varphi{_{\overline{V_i}}} \wedge \neg \varphi) \wedge \bigwedge_{j=1}^r \alw(z_i \leftrightarrow z_i' )  \wedge \alw(z \leftrightarrow z')
$$
is unsatisfiable.
This contradicts Lemma~\ref{lem:escape}.
Hence the variable $z$ added at stage $i$ must belong to $U$, and therefore
$V_{i+1} \subseteq U$.

We conclude that $V_i \subseteq U$ for every $i$, and in particular
$V=V_k \subseteq U$, contradicting $U \subsetneq V$.
\end{proof}

\begin{corollary}
The algorithm ${\cal NDC}$ returns a partition of $\Y$ into
minimal independent sets.
\end{corollary}

\begin{proof}
At each iteration of the main loop, ${\cal NDC}$ computes a set $V \subseteq \Y$ and adds it to the output partition. By Theorem~\ref{thm:look}, each such set is independent in $\varphi$ and has no non-empty proper independent subset.

Moreover, after adding $V$ to the partition, the algorithm removes all variables in $V$ from the set of remaining variables. Hence the sets computed in different iterations are pairwise disjoint. Since the loop stops only when no variables remain, the union of all computed sets is exactly $\Y$.
Therefore, ${\cal NDC}$ returns a partition of $\Y$ into minimal independent sets.
\end{proof}

\section{Conclusion}
This paper has analysed the decomposition algorithm ${\cal DC}$ from a semantic perspective and used that analysis to derive a complete refinement, namely ${\cal NDC}$. Our results are particularly relevant to the area of \ltl-based Assume/Guarantee contracts, where ${\cal DC}$ is used to decompose specifications into smaller subspecifications.

It is important to emphasise that all developments in this paper rely on the notion of independence adopted in~\citep{I18,I24}, which is formulated in terms of traces. This is natural in the satisfiability setting, where traces are the semantic objects used to interpret formulae. By contrast, in the setting of realisability, the relevant semantic objects are winning strategies rather than traces. This distinction is significant when one considers whether analogous results can be obtained in the realisability setting.

To reason about realisability and synthesis, one must work with winning strategies rather than with traces alone. Let
$$
\Sigma^{\cal W}_\varphi
=
\{\sigma^{\theta,E} \mid \theta \text{ is a winning strategy for } \varphi,\ E \in (2^\E)^\omega\}
$$
be the set of traces induced by winning strategies for $\varphi$. By Definition~\ref{def:winstr}, every such trace is a model of $\varphi$, and therefore
$
\Sigma^{\cal W}_\varphi \subseteq \Sigma_\varphi.
$
However, the converse inclusion fails: there are traces that satisfy $\varphi$ but cannot arise from any winning strategy. For instance, consider the formula 
from Examples~\ref{ex:numberzero} and~\ref{ex:last}:
$
\varphi = \alw\bigl((\neg p \to \alw \neg b) \wedge (p \to (\alw a \vee \alw b))\bigr),
$
where $p$ is the only environment variable. The trace $\{p,b\}^\omega$ belongs to $\Sigma_\varphi$, but it is not induced by any winning strategy for $\varphi$ (see Example~\ref{ex:numberzero}). Thus,
$
\Sigma_\varphi \not\subseteq \Sigma^{\cal W}_\varphi.
$

This observation raises two questions. The first is whether, when considering the realisability problem rather than the satisfiability problem, ${\cal NDC}$  produces a partition into sets of independent variables?  The second is whether these subsets are minimal, that is, whether they cannot be further decomposed while preserving their independence. The first question can be answered affirmatively, since $\Sigma^{\cal W}_{\varphi} \subseteq \Sigma\varphi$. The second, however, is considerably more challenging, unlikely to be simpler than the underlying realisability problem.

In summary, identifying independent variables in reactive systems with the aid of a model checker is a very interesting approach, as it can uncover sets of variables that are effectively independent at a computational cost lower than that of solving the full realizability problem. 
However, 
extending these ideas to strategy-based independence would require reasoning about winning strategies themselves. Consequently, it is unlikely that such notions can be computed substantially more efficiently than realizability itself.

\bibliographystyle{elsarticle-harv} 
\bibliography{references}



\end{document}

%% file: commands.tex
\newtheorem{lemma}{Lemma}
\newtheorem{theorem}{Theorem}
\newtheorem{definition}{Definition}
\newtheorem{example}{Example}
\newtheorem{corollary}{Corollary}

\newcommand{\KWD}[1]{\ensuremath{\textit{#1}}\xspace}

\newcommand{\Not}{\mathclose{\neg}}

\newcommand{\prop}{\mathbf{\sf Prop}}

\newcommand{\cF}{{\bf\scriptstyle F}}
\newcommand{\cT}{{\bf\scriptstyle T}}
\newcommand{\nx}{\mathlarger{\fullmoon}}
\newcommand{\ev}{\mathlarger{\diamondsuit}}
\newcommand{\alw}{\square}
\newcommand{\until}{\,\mathcal{U}\,}
\newcommand{\releases}{\mathcal{R}}

\newcommand{\res}{\!\upharpoonright\!}

\makeatletter
\def\BState{\State\hskip-\ALG@thistlm}
\makeatother

\newcommand{\E}{\mathcal{I}}
\newcommand{\Y}{\mathcal{O}}

\newcommand{\ltl}{\KWD{\sf LTL}}
\newcommand{\Iff}{\mathrel{\leftrightarrow}}